\documentclass[showpacs,twocolumn,pra,notitlepage,floatfix]{revtex4-1}

\usepackage{amsmath,amssymb,amsthm,mathrsfs,amsfonts,dsfont}
\usepackage{bm}
\usepackage{enumerate}
\usepackage{diagbox}
\usepackage{physics}
\usepackage{color}
\usepackage{multirow}
\usepackage{comment}
\usepackage{tikz}
\usetikzlibrary{external}

\usepackage[colorlinks = true]{hyperref}

\newtheorem{theorem}{Theorem}

\newtheorem{lemma}{Lemma}
\newtheorem{corollary}{Corollary}

\newtheorem{definition}{Definition}

\graphicspath{{./figure/}}

\begin{document}

\title{Fundamental Limitation on the Detectability of Entanglement}

\author{Pengyu Liu}
\affiliation{Center for Quantum Information, Institute for Interdisciplinary Information Sciences, Tsinghua University, Beijing 100084, China}

\author{Zhenhuan Liu}
\affiliation{Center for Quantum Information, Institute for Interdisciplinary Information Sciences, Tsinghua University, Beijing 100084, China}

\author{Shu Chen}
\affiliation{Center for Quantum Information, Institute for Interdisciplinary Information Sciences, Tsinghua University, Beijing 100084, China}

\author{Xiongfeng Ma}
\email{xma@tsinghua.edu.cn}
\affiliation{Center for Quantum Information, Institute for Interdisciplinary Information Sciences, Tsinghua University, Beijing 100084, China}

\begin{abstract}
Entanglement detection is essential in quantum information science and quantum many-body physics. It has been proved that entanglement exists almost surely for a random quantum state, while the realizations of effective entanglement criteria usually consume exponentially many resources with regard to system size or qubit number, and efficient criteria often perform poorly without prior knowledge. This fact implies a fundamental limitation might exist in the detectability of entanglement. In this work, we formalize this limitation as a fundamental trade-off between the efficiency and effectiveness of entanglement criteria via a systematic method to evaluate the detection capability of entanglement criteria theoretically. For a system coupled to an environment, we prove that any entanglement criterion needs exponentially many observables to detect the entanglement effectively when restricted to single-copy operations. Otherwise, the detection capability of the criterion will decay double exponentially. Furthermore, if multicopy joint measurements are allowed, the effectiveness of entanglement detection can be exponentially improved, which implies a quantum advantage in entanglement detection problems. Our results may shed light on why quantum phenomena are difficult to observe in large noisy systems.
\end{abstract}

\date{\today}
\maketitle

Quantum information technology promises advancement in various information processing tasks. Currently, we are in a stage where noisy intermediate-scale quantum devices~\cite{Preskill2018quantumcomputingin} with 50 to 200 qubits can be well manipulated to demonstrate quantum advantages~\cite{arute2019quantum, Gong2021zuchongzhi,zhong2021jiuzhang,madsen2022quantum}. For these devices, entanglement generation is regarded as an important benchmark, while the verification of systems with only 18 qubits is already challenging~\cite{wang2018eighteen}. This is rather counterintuitive as entangled states have been proved to constitute a large proportion of state space~\cite{zyczkowski1998volume,szarek2005volume,gurvits2002largest}, even for highly mixed states~\cite{aubrun2012phase}. 

Among the various detection methods, entanglement witness (EW) criteria are rather straightforward and the most commonly used ones in experiments~\cite{lu2018structure,wang2018eighteen}. However, much evidence shows that the EW criteria are only effective with precise prior knowledge of the target state~\cite{nidari2007witness}. Unpredictable noises in the state preparation could significantly reduce the success probability for EW protocols.

To solve this problem, researchers have developed nonlinear entanglement criteria, such as positive map criteria, including the well-known positive partial transposition (PPT) criterion~\cite{peres1996separability}, computable cross norm or realignment (CCNR) criterion~\cite{chen2002matrix}, and symmetric extension criterion~\cite{GUHNE2009detection}. Although more effective than EW criterion, checking these nonlinear criteria relies heavily on state tomography, which is experimentally unaffordable. In the last few decades, many efforts have been devoted to modifying these powerful entanglement criteria, such as the positive map criteria, to avoid state tomographies~\cite{horodecki2002method,horodecki2003measuring}. 

With the intermediate-scale quantum devices available, entanglement criteria have been applied to various physical systems. For these experiments, the experimental feasibility --- low sample complexity and single-copy compatibility --- becomes a growing concern for criterion design. Protocols like the moment-based PPT and CCNR criteria~\cite{elben2020mixed,yu2021optimal,neven2021symmetry,liu2022detecting} are proposed which can even be realized by single-copy and qubit-wise measurements when combined with the randomized measurements techniques~\cite{van2012measuring,huang2020predicting,brydges2019probing}. Although much more efficient than state tomography, these methods still require a number of measurements that scales exponentially with the system size. In addition to EW and moment-based criteria, many other case studies investigating the detection capability of some specific entanglement criteria~\cite{lu2016universal,collins2016random,bhosale2012entanglement,shapourian2021diagram,jivulescu2014reduction,jivulescu2015thresholds,aubrun2012realigning} also suggest that a trade-off may exist between the effectiveness and the efficiency of entanglement detection. However, a general and quantitative study is still missing.

In this work, we develop a systematic method to upper bound the detection capability of various entanglement criteria, including EW, positive map, and faithful entanglement criteria. We further generalize it to any entanglement criteria with single-copy implementations and theoretically formulate the fundamental trade-off between efficiency and effectiveness, see Theorem \ref{theorem:singlecopy}. Here we give an informal version.
\begin{theorem}[Trade-off between Efficiency and Effectiveness, Informal]
To detect the entanglement of a random state coupled to a $k$-dimensional environment, any entanglement criterion that can be verified experimentally with $M$ observables is either 
\begin{enumerate}
    \item Inefficient: The criterion requires $M=\Omega(k/\ln k)$ observables to verify, or 
    \item Ineffective: The criterion can detect the entanglement successfully with a probability $P=e^{-\Omega(k)}$ even if the state is entangled.
\end{enumerate}
\label{theorem:informal}
\end{theorem}
Explicitly speaking, we investigate the entanglement within a bipartite system $AB$, and system $R$ is their purification with dimension $k$. The composite system $ABR$ as a whole is in a random pure state. System $R$ can be regarded as the environment of $AB$, representing either the uncontrollable noise or some system that is not of concern. Such a composite system $ABR$ often appears in many-body physics as it can be generated by a generic Hamiltonian. Note that $k$ usually scales exponentially with the environment size. So, according to Theorem \ref{theorem:informal}, the number of observables increases exponentially, and the detection capability decreases double exponentially with the environment size.

To formalize our study quantitatively, here we give a formal definition of density state distribution~\cite{collins2016random, Nechita2007}.

\begin{definition}[$k$-induced Distribution of Density Matrix]
$\pi_{d,k}$ is the distribution in $\mathcal{D}(\mathcal{H})$ induced by the uniform distribution of pure states in $\mathcal{H}\otimes\mathcal{H}_R$, where the dimensions of $\mathcal{H}$ and $\mathcal{H}_R$ are $d$ and $k$ respectively. A state $\rho$ following the distribution $\pi_{d,k}$ can be generated by $\rho=\tr_R(\ketbra{\phi})$, where $\ket\phi$ is a Haar-measured pure state in $\mathcal{H}\otimes\mathcal{H}_R$.
\end{definition}

Let us start with EW criteria. An EW is an observable, $W$, satisfying $\tr(W\rho)\geq 0,\forall \rho\in \mathrm{SEP}$ where $\mathrm{SEP}$ is the set of all separable states. Define the detection capability of an EW criterion with $W$ as
\begin{equation}
    \mathcal{C}_k(W)=\Pr_{\rho\sim\pi_{d,k}}\bqty{\tr(W\rho)<0},
\end{equation}
which represents the portion of states that $W$ can detect. Without loss of generality, hereafter, we assume the two subsystems $A$ and $B$ are equal in dimension, $d_A=d_B=\sqrt{d}$. It has been proved that when $k<cd^{\frac{3}{2}}$, where $c$ is some constant, a state following $\pi_{d,k}$ distribution is entangled with probability $1$ asymptotically~\cite{aubrun2012phase}. Throughout the Letter, we will always assume $k<cd^{\frac{3}{2}}$ so that the definition of $\mathcal{C}_k(W)$ can also be viewed as the ratio of detected states to all entangled states.

Using Laurent-Massart's lemma~\cite{laurent2000adaptive}, we can give an upper bound of the detection capability of EW criteria.
\begin{theorem}[Detection Capability of EW Criteria]\label{theorem:maintheorem}
The detection capability of an EW criterion with $W$ decays at least exponentially with the dimension of the environment
\begin{equation}
    \mathcal{C}_k(W)< 2e^{-(\sqrt{1+\alpha}-1)^2 k}\leq 2e^{-(3-2\sqrt{2}) k},
\end{equation}
where $\alpha=\frac{\tr(W)}{\sqrt{\tr(W^2)}}\geq 1$~\cite{JOHNSTON20181} is a witness-dependent factor.
\end{theorem}

We show the proof of Theorem~\ref{theorem:maintheorem} intuitively in Fig.~\ref{fig:geometry_new}. When $k$ is large, the state distribution $\pi_{d,k}$ converges near the surface of the set of separable states. An entanglement witness can only detect states in a high-dimensional spherical cap due to the constraint of $\tr(W\rho)\geq 0,\forall \rho\in \mathrm{SEP}$. Since a spherical cap in high-dimensional space is exponentially small compared to the ball, $\mathcal{C}_k(W)$ also suffers from an exponential decay. Detailed proofs of this theorem and the rest can be found in the Appendix.

\begin{figure}[ht]
    \centering
    \includegraphics[width=8cm]{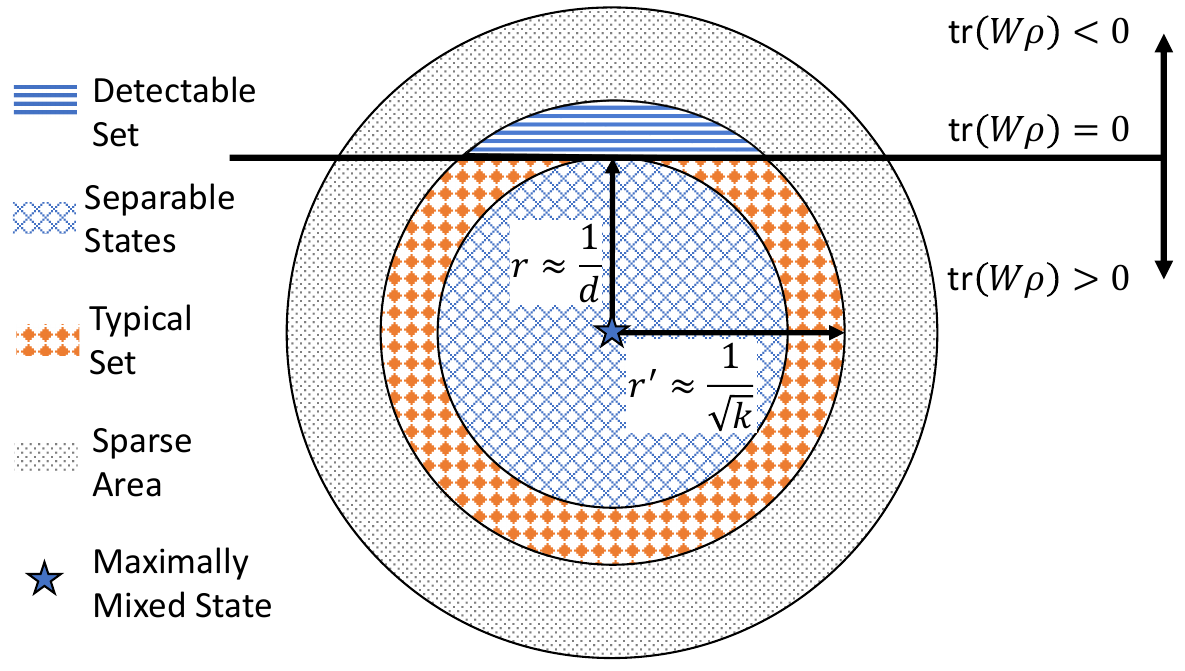}
    \caption{
An intuitive illustration of Theorem~\ref{theorem:maintheorem}. We use three balls in the $(d^2-1)$-dimensional space, $\mathcal{D}(\mathcal{H}_d)$, to represent the states where the outer one contains all the states. The maximally mixed state has purity $\frac{1}{d}$, which lies in the middle of the figure. It is also shown that all states with distance to the maximally mixed state smaller than $r\approx \frac{1}{d}$ are separable~\cite{gurvits2002largest}, which is represented as the inner ball. With state distribution $\pi_{d,k}$, the expected purity can be approximated by $\frac{1}{d}+\frac{1}{k}$~\cite{holgersson2020recent}. As a result, the state distribution will concentrate in the typical set with $r'\approx \frac{1}{\sqrt{k}}$ centered at the maximally mixed state as the middle ball. Outside the typical set is the sparse area, where we can ignore the existence of states. The horizontal line represents the hyperplane defined by $\tr(W\rho)=0$. The states above the hyperplane satisfying $\tr(W\rho)<0$ are detectable by $W$, which forms a high-dimensional spherical cap. We can approximate $\pi_{d,k}$ with a uniform distribution inside the typical set. The detection capability of an EW is bounded by the volume ratio of the detectable set, which is exponentially small and bounded by an order of $e^{-(d^2-2)(\frac{r}{r'})^2/2}\sim e^{-\frac{k}{2}}$~\cite{blum2020foundations}.
    }
    \label{fig:geometry_new}
\end{figure}

This theorem explains why the effectiveness of EW criteria highly depends on the prior knowledge of the studied states, as the detection capability decreases double-exponentially fast with the environment size. It is also worth mentioning that this result holds for multipartite EWs and the leftmost inequality holds for any observable $O$ with a positive trace.

We use two typical examples to support our results. The first example is PPT-type EW, $W=\ketbra{\phi}^{T_A}$, where $T_A$ is the partial transposition operator acting on $\mathcal{H}_A$ and $\ket\phi$ is an arbitrary pure state. In the sense of detection capability, they are optimal EWs as $\alpha$ achieves its minimum value, $\alpha=\frac{\tr(W)}{\sqrt{\tr(W^2)}}=1$, which is irrelevant with the system dimension, $d$. Hence, we have
\begin{equation}
\mathcal{C}_k(\ketbra{\phi}^{T_A})= e^{-\Omega(k)}.
\end{equation} 
In fact, this inequality is rather tight as there exists a constant $c$ such that $\mathcal{C}_k(\ketbra{\phi}^{T_A})\geq e^{-ck}$
according to Ref.~\cite{nidari2007witness}.

The second example is the faithful EW, defined as $W=\frac{\mathbb{I}}{\sqrt{d}}-\ketbra{\Phi}{\Phi}$, where $\mathbb{I}$ is the identity operator and $\ket{\Phi}$ is a maximally entangled state in $\mathcal{H}_A\otimes\mathcal{H}_B$. Such kinds of fidelity-based EWs are commonly used in practical entanglement detection tasks~\cite{wang2018eighteen} as many efficient fidelity estimation protocols exist~\cite{huang2020predicting,flammia2011direct}. However, Theorem \ref{theorem:maintheorem} tells us that such an entanglement witness performs extremely weak in the sense that its detection capability also decreases with system size since $\alpha=\frac{\tr(W)}{\sqrt{\tr(W^2)}}=\sqrt{\frac{d-\sqrt{d}}{2}}$. As a result, 
\begin{equation}
\mathcal{C}_k\left(\frac{\mathbb{I}}{\sqrt{d}}-\ketbra{\Phi}\right)=e^{-\Omega(\sqrt{d}k)}.
\end{equation}

To make our results more convincing, we conduct several numerical experiments, as shown in Fig.~\ref{fig:detectioncap}. We generate random states according to distribution $\pi_{d,k}$ with different values of $d$ and $k$ and use the two kinds of EWs discussed above to detect it. From Fig.~\ref{fig:detectioncap}(a), one could find that the detection capabilities of all types of EWs exponentially decay with $k$. Besides, the slopes of the faithful EW with $d=4$ and two PPT EWs are almost the same, which fulfills the prediction of Theorem \ref{theorem:maintheorem} as $\alpha=1$ for these three EWs. The slope of the faithful EW with $d=9$ is smaller than the other three EWs, reflecting that the value of $\alpha$ for faithful EWs increases with system dimension. In Fig.~\ref{fig:detectioncap}(b), we investigate the relation between detection capability and system dimension. One could find that the detection capabilities of PPT-type EWs have no apparent changes when increasing the system dimension. In comparison, the detection capability of faithful EWs shows exponential decaying behavior, and the slopes decrease as $k$ increases. These phenomena all satisfy our predictions.

\begin{figure}[htbp]
\centering
\includegraphics[width=8.5cm]{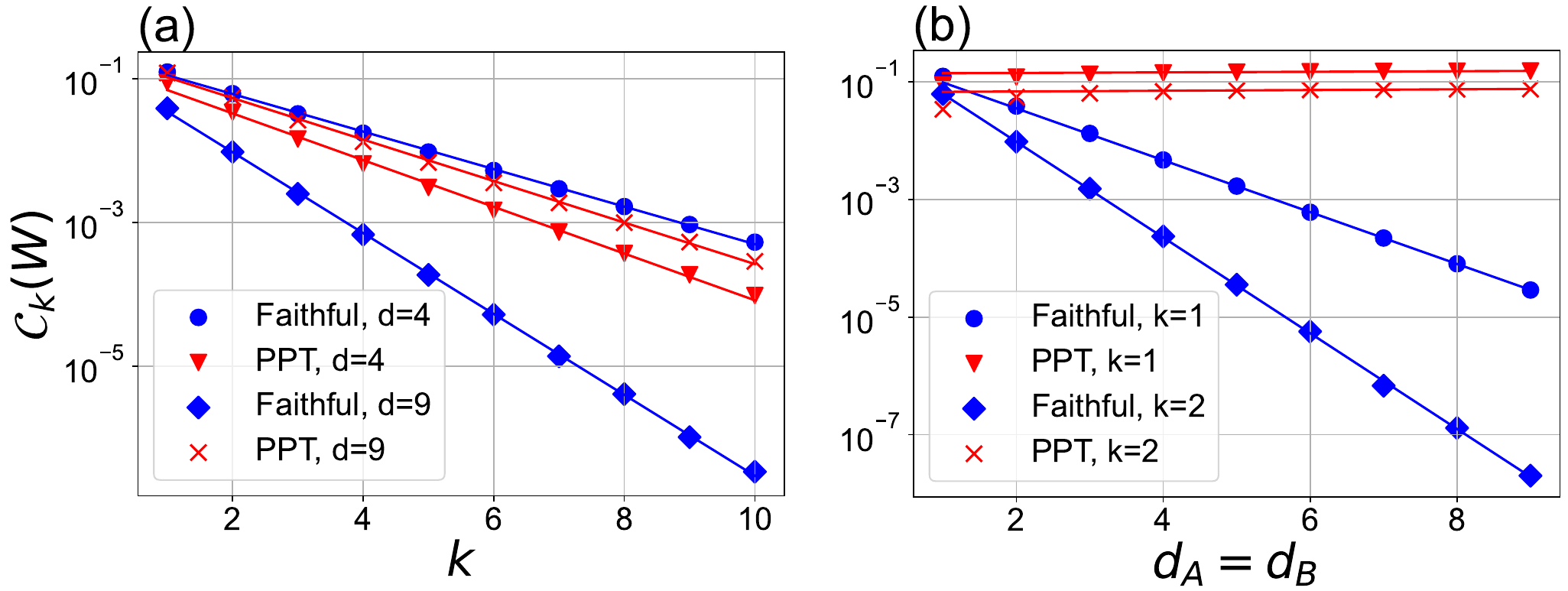}
\caption{Scaling of detection capability of EW criteria with
regard to (a) the environment dimension $k$ and (b) the system
dimension $d_A=d_B$. To numerically calculate the detection capability, we generate $10^8$ density matrices following $\pi_{d,k}$ for each point and treat them as bipartite states $\rho_{AB}$ with $d_A=d_B=\sqrt{d}$. For each randomly generated state, we use two kinds of EWs to detect it: the PPT type EWs $W=\ketbra{\phi}^{T_A}$ and the faithful EWs $W=\frac{\mathbb{I}}{\sqrt{d}}-\ketbra{\psi}$. Here $\ket{\phi}$ is a random state and $\ket{\psi}=U_A\otimes U_B\sum_{i=1}^{\sqrt{d}}\ket{ii}$ is a random maximally entangled state, where $U_A$ and $U_B$ follow Haar-measure distribution for each randomly sampled state $\rho_{AB}$. The straight lines are linear regression results with absolute slopes all larger than the slopes predicted by Theorem~\ref{theorem:maintheorem}.}
    \label{fig:detectioncap}
\end{figure}
Since EW criteria highly depend on prior knowledge to succeed, a direct improvement is to combine a large number of EWs. Naturally, we define an EW set $\mathcal{W}=\Bqty{W_i,i=1\cdots N}$ and the corresponding detection capability as
\begin{equation}
    \mathcal{C}_k(\mathcal{W})=\Pr_{\rho\sim\pi_{d,k}}\bqty{\exists W\in \mathcal{W}: \tr(W\rho)<0}.
\end{equation}
By using the union bound, we can show that the detection
capability of the finite EW set still decreases exponentially
when $k$ is large:
\begin{equation}\label{eq:manyEWs}
    \mathcal{C}_k\pqty{\mathcal{W}}< 2N{e}^{-(\sqrt{1+\alpha_{\min}}-1)^2 k}\le 2{e}^{\ln(N)-(3-2\sqrt{2}) k}
\end{equation}
where $\alpha_{\min}=\min_{W\in \mathcal{W}}\frac{\tr(W)}{\sqrt{\tr(W^2)}}\geq 1$. Therefore, to effectively detect entanglement, a total number of ${e}^{\Omega(k)}$ EWs is required, which is extremely impractical.

There are many other theoretically attractive entanglement criteria and concepts based on EWs. Examples like the positive map criteria~\cite{GUHNE2009detection} and faithful entanglement~\cite{weilenmann2020faithful,gunhe2021geometry,riccardi2021exploring} are equivalent to infinitely many EWs. As a result, Eq.~\eqref{eq:manyEWs} does not apply directly. To adapt the previous theorem to the infinite case, here we define parameterized EW criteria. 
\begin{definition}[Parameterized EW Criteria]\label{def:paraEW}
A parameterized EW criterion is a set of an infinite number of EWs, which can be represented by a map $\mathcal{M}$ from $M$ real parameters to EWs in $\mathcal{D}(\mathcal{H})$, satisfying
\begin{equation}
    \forall \vb* \theta\in\Theta\subset[-1,1]^M,\forall\rho\in \mathrm{SEP}:\tr\left[\rho \mathcal{M}(\vb* \theta)\right]\geq 0,
\end{equation}
where $\mathcal{M}(\vb*\theta)$ is a normalized EW satisfying $\norm{\mathcal{M}(\theta)}_F=1$ with $\norm{A}_F=\sqrt{\sum_{i,j}|A_{i,j}|^2}$ being the Frobenius norm and $\Theta$ is the feasible parameter space ensuring $\mathcal{M}(\vb* \theta)$ a valid EW. A state $\rho$ can be detected by this criterion if and only if 
\begin{equation}
    \exists \vb*\theta\in\Theta: \tr\pqty{\rho \mathcal{M}(\vb* \theta)}<0.
\end{equation}
\end{definition}

Similarly, we can define the detection capability of a parameterized EW as
\begin{equation}
    \mathcal{C}^p_k(\mathcal{M})=\Pr_{\rho\sim\pi_{d,k}}\bqty{\exists \vb*\theta\in\Theta: \tr(\rho \mathcal{M}(\vb* \theta))<0}.
\end{equation}
By using a coarse-graining method and adopting Theorem \ref{theorem:maintheorem}, we can derive an upper bound for $\mathcal{C}^p_k(\mathcal{M})$.

\begin{theorem}[Detection Capability of Parameterized EW Criteria]
For any parameterized EW represented by a normalized $l$-Lipschitz map $\mathcal{M}$ satisfying
\begin{equation}
\forall \vb*\theta,\vb*\theta'\in \Theta:\norm{\mathcal{M}(\vb*\theta)-\mathcal{M}(\vb*\theta')}_F\leq l\norm{\vb*\theta-\vb*\theta'}_2,
\end{equation}
the detection capability decays at least exponentially with $k$ after $k$ exceeds a certain threshold,
\begin{equation}
    \mathcal{C}^p_k(\mathcal{M})< 2e^{C_1-C_2k},
\end{equation}
where $C_1=M\ln 4\sqrt{M}ld$, $M$ is the number of real parameters in $\mathcal{M}$, $C_2=(\sqrt{0.5+\alpha_{\min}}-1)^2$ where $\alpha_{\min}=\min_{\vb*\theta} \frac{\tr[\mathcal{M}(\vb*\theta)]}{\sqrt{\tr[\mathcal{M}(\vb*\theta)^2]}}=\min_{\vb*\theta} \tr[\mathcal{M}(\vb*\theta)]\geq 1$.
\label{theorem:paraEWs}
\end{theorem}

The definition of a parameterized EW criterion naturally covers positive map criteria. If a state $\rho$ does not satisfy $\mathcal{N}_A\otimes\mathbb{I}_B(\rho)\geq 0$ for a positive map $\mathcal{N}$, then $\exists \ket\phi: \tr\left[\rho\mathcal{N}_A\otimes\mathbb{I}_B(\ketbra{\phi})\right]<0$. Regarding $\ket{\phi}$ as the parameters $\vb* \theta$ in theorem~\ref{theorem:paraEWs}, this theorem can be applied directly. We leave the detailed discussion in the Appendix.

Another example of parameterized EW is the faithful entanglement, proposed in~\cite{weilenmann2020faithful}, which refers to those entangled states detected by faithful EWs as defined before. We define a parameterized EW that is equivalent to all the faithful EWs as $\mathcal{M}_{\mathrm{faithful}}(\vb*\theta)=\pqty{\sqrt{2-\frac{2}{\sqrt{d}}}}^{-1}\pqty{\frac{\mathbb{I}}{\sqrt{d}}-\ketbra{\phi(\vb*\theta)}}$, where $\ket{\phi(\vb*\theta)}$ is a maximally entangled state~\cite{gunhe2021geometry}. One could prove that $\mathcal{M}_{\mathrm{faithful}}(\vb*\theta)$ is at least $\sqrt{2}$-Lipschitz and $\alpha_{\mathrm{min}}=\sqrt{\frac{d-\sqrt{d}}{2}}\approx\sqrt{\frac{d}{2}}$ when $d$ is large. So that an upper bound for the ratio of faithful entangled states can be summarized below using Theorem \ref{theorem:paraEWs}.

\begin{corollary}[Ratio of Faithful Entanglement States]
The set of faithful entangled states has an exponentially small ratio in the state space:
\begin{equation}
        \Pr_{\rho\sim \pi_{d,k}}[\rho\in \mathrm{FE}]=\mathcal{C}^p_k\pqty{\mathcal{M}_{\mathrm{faithful}}}< 2e^{C_1-C_2k}
\end{equation}
where $\mathrm{FE}$ is the set of all faithful entangled states and $C_1=3d\ln 4d$, $C_2=\pqty{\sqrt{0.5+\sqrt{\frac{d-\sqrt{d}}{2}}}-1}^2\approx \sqrt{\frac{d}{2}}$.
\end{corollary}

This result shows when $k=\Omega(\sqrt{d}\ln d)$, the faithful EWs can hardly detect entanglement, which is compatible with the numerical results shown in Ref.~\cite{gunhe2021geometry}.

Besides positive map and faithful criteria, there are many other entanglement criteria designed for different scenarios, like the one based on the state moments~\cite{imai2021bound,elben2020mixed,liu2022detecting}, uncertainty relations~\cite{duan2000inseparability,gunhe2004uncertainty}, and machine learning~\cite{gray2018machine,yin2022efficient}. They may use complex mathematical relations and complicated postprocessing to detect the entanglement. While limited by the basic principles of quantum mechanics and current technology, only values like $\tr(O\rho)$ can be measured directly. Hence, we propose a general definition of entanglement criteria with single-copy realizations.

\begin{definition}[Single-Copy Criteria]\label{def:singlecopy}
An entanglement criterion is said to have a single-copy realization if it can be checked by the expectation of a set of observables $\mathcal{O}=\Bqty{O_i|i=1,\cdots, M}$. 
After the measurement, one gets the results, $r_{\rho,i}=\tr(O_i\rho), i=1,\cdots,M$, and can decide the feasible region $F_\mathcal{O}(\rho)$ of the state
\begin{equation}
    F_\mathcal{O}(\rho)=\Bqty{\sigma\in\mathcal{D}(\mathcal{H}_d)|\tr(O_i\sigma)=r_{\rho,i},i=1,\cdots,M}.
\end{equation}
If 
\begin{equation}
    F_{\mathcal{O}}(\rho)\cap \mathrm{SEP}=\varnothing\label{eq:singlecopycriterion},
\end{equation}then $\rho$ is entangled.
\end{definition}

According to this definition, we can define the detection capability of the single-copy criterion $\mathcal{O}$ as
\begin{equation}\label{eq:singlecopycap}
    \mathcal{C}_k^s(\mathcal{O})=\Pr_{\rho\sim\pi_{d,k}}\bqty{F_{\mathcal{O}}(\rho)\cap \mathrm{SEP}=\varnothing}.
\end{equation}
Since the verification of Eq.~\eqref{eq:singlecopycriterion} might require exponentially many classical resources, many practical entanglement criteria are essentially designed by finding supersets of $\mathrm{SEP}$ and $F_{\mathcal{O}}(\rho)$ and deciding whether these two supersets are disjoint or not. Therefore, the previous definition is the strongest criterion using the measurement results of $\mathcal{O}$, and Eq.~\eqref{eq:singlecopycap} gives an upper bound for all criteria using the same data.

Without loss of generality, we could assume that all the observables are mutually orthogonal and normalized, i.e., $\tr(O_iO_j)=\delta_{i,j}$. In the Appendix, we prove that if a state $\rho$ can be detected by a single-copy criterion $\mathcal{O}$ which contains $M-1$ observables, then it can be detected by a $1$-Lipschitz parameterized EW with $M$ parameters.
\begin{equation}
\mathcal{M}(\theta_0,\theta_1,\cdots,\theta_{M-1})=\theta_0\mathbb{I}+\sum_{i=1}^{M-1}\theta_iO_i
\end{equation}
Hence, directly adopting Theorem \ref{theorem:paraEWs}, one can give an upper bound for $\mathcal{C}_k^s(\mathcal{O})$.

\begin{theorem}[Detection Capability of Single-Copy Criteria]\label{theorem:singlecopy}
Any single-copy entanglement criterion $\mathcal{O}$ with $M-1$ observables has detection capability
\begin{equation}
    \mathcal{C}^s_k(\mathcal{O})< 2e^{C_1-C_2k},
\end{equation}
where $C_1=M\ln 4\sqrt{M}d$, $C_2=(\sqrt{1.5}-1)^2\approx 0.05$. 
\end{theorem}

Theorem \ref{theorem:singlecopy} theoretically formulates the trade-off between the effectiveness and sample complexities of entanglement criteria. According to this theorem, at least $\Omega(\frac{k}{\ln k})$ observables are needed to effectively detect the entanglement of a random state, even assuming the measurement results are infinitely accurate. Besides, compared with Eq.~\eqref{eq:manyEWs}, one could conclude that a general single-copy detection can be exponentially better than simply using a set of EWs. 

Here, we numerically examine the detection capabilities of several nonlinear criteria, like purity~\cite{GUHNE2009detection}, fisher information~\cite{Zhang_2020}, moments of partial transposed~\cite{yu2021optimal,neven2021symmetry} (labeled by $D_{3,\mathrm{opt}}$) and realigned density matrices~\cite{liu2022detecting} (labeled by $M_4$). We leave the description of these four criteria for the Appendix. These criteria all have single-copy realizations with resources independent of $k$. Therefore, from Fig.~\ref{fig:single-copy}, one could find that the detection capabilities of these four criteria decay exponentially with $k$ when $k$ is large, which is compatible with Theorem \ref{theorem:singlecopy}. 
\begin{figure}[htbp]
\centering
\includegraphics[width=8cm]{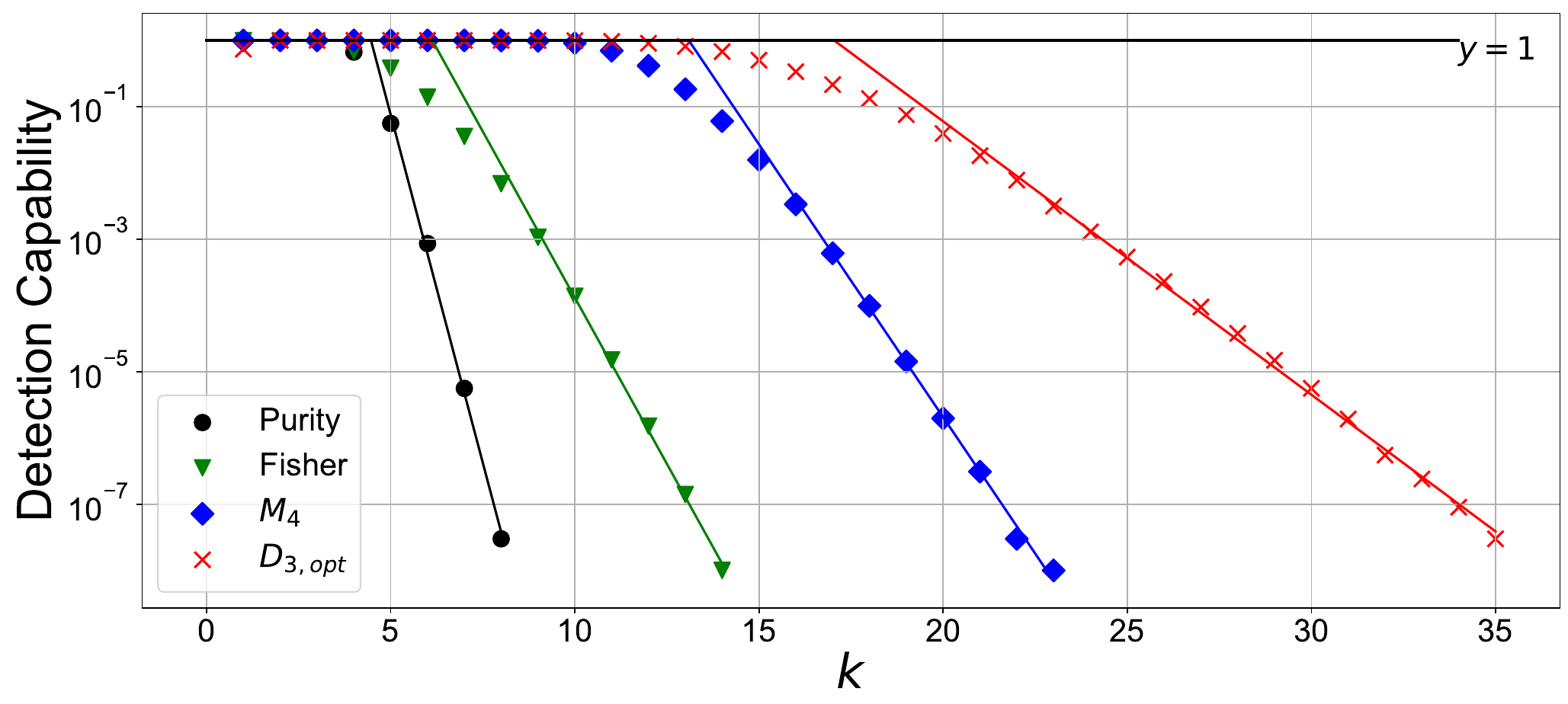}
    \caption{Detection capability of four nonlinear criteria. Here for each point, we generate $10^8$ states $\rho\in\mathcal{D}(\mathcal{H}_4\otimes\mathcal{H}_4)$ with distribution $\pi_{16,k}$ and calculate the detection capability of four different nonlinear entanglement criteria. The four inclined straight lines are linear regression results on the last several points with absolute slopes all larger than the ones predicted by Theorem~\ref{theorem:singlecopy}.}
    \label{fig:single-copy}
\end{figure}

Before the exponential decaying period, we also observe that the detection capabilities hold constant. In the Appendix, we analyze these thresholds in detail and numerically find that they all have polynomial relations with the system dimension $d$. Like for the $D_{3,\mathrm{opt}}$ criterion, the threshold is linearly dependent on $d$. These observations together with Theorem \ref{theorem:singlecopy} explain why the verification of these four criteria needs exponentially many  resources~\cite{brydges2019probing,rath2021fisher,zhou2020single,elben2020mixed,liu2022detecting}. From another point of view, if not restricted to single-copy operations, some of these criteria can be realized by only a few multicopy observables, implying a quantum advantage in entanglement detection tasks by joint operations~\cite{Huang2022}.

We can prove this advantage in some special cases. Let $d_A=d_B=k=\sqrt{d}$, the distributions of $\tr(\rho_A^2)$ and $\tr(\rho_{AB}^2)=\tr(\rho_R^2)$ are completely the same as systems $A$ and $R$ are symmetric. Hence, using the purity criterion, i.e. $\tr(\rho_{AB}^2)\le\tr(\rho_A^2) \ \forall \rho\in \mathrm{SEP}$, the detection capability is $0.5$ and the criterion can be verified using just one two-copy observable, $\tr(\rho_{AB}^2)-\tr(\rho_A^2)=\tr[(\mathbb{S}_{AB}-\mathbb{S}_A)\rho_{AB}^{\otimes 2}]$, where $\mathbb{S}$ is the SWAP operator. So we can summarize the results below.

\begin{corollary}[Quantum Advantage in Entanglement Detection]
Consider a state following $\pi_{d,\sqrt{d}}$ distribution, and $d_A=d_B=\sqrt{d}$. With only single-copy measurements, $M=\Omega(\frac{\sqrt{d}}{\ln d})$ observables are required for any criterion with detection capability greater than $0.5$. However, if multicopy joint measurements are allowed, one can detect with a capability equaling $0.5$ with only one two-copy observable.
\label{proposition:advantage}
\end{corollary}

Beyond Definition \ref{def:singlecopy}, adaptive methods could also be used to increase the efficiency of entanglement detection. In the Appendix, we give similar results as Theorem \ref{theorem:singlecopy} and Corollary \ref{proposition:advantage} for adaptive methods. It should be noticed that the quantum advantage in Corollary \ref{proposition:advantage} only holds in terms of the number of observables. While considering real-world experiments where multicopy measurements may require much more resources than single-copy ones, will the advantage still hold soundly? Besides, will Theorem \ref{theorem:singlecopy} holds when a small false-positive error rate is allowed? We will leave these questions to future work.

Meanwhile, our result also holds for some other typical state distributions. For example, we can show that Theorem~\ref{theorem:singlecopy} applies to random thermal states, which is widely used in quantum thermodynamics~\cite{vinjanampathy2016quantum}. In the Appendix, we present some numerical results demonstrating the exponential decay behavior of detection capabilities for random thermal states.

\begin{acknowledgments}
We thank Zhaohui Wei for the valuable discussions. This work was supported by the National Natural Science Foundation of China Grants No.~11875173 and No.~12174216 and the National Key Research and Development Program of China Grants No.~2019QY0702 and No.~2017YFA0303903. 
\end{acknowledgments}

\bibliographystyle{apsrev4-1}
%

\appendix

\onecolumngrid
\newpage

\clearpage
\setcounter{theorem}{1}

\section{Detection Capability Upper Bound of EW Criteria}\label{section: EW}
\subsection{Restriction of Valid EWs}
\begin{lemma}[Restriction of Valid EWs]
For any valid EW $W$ satisfying 
\begin{equation}
\forall \rho\in \mathrm{SEP}: \tr(W\rho)\geq 0,
\end{equation}
the following inequality always holds
\begin{equation}
    \alpha^2=\frac{\tr(W)^2}{\tr(W^2)}\geq 1.
\end{equation}
\end{lemma}
\begin{proof}
Given an EW $W$, without loss of generality, we assume $\tr(W)=1$. Write $W$ in the form
\begin{equation}
    W=\frac{\mathbb{I}}{d}+\frac{c\sigma}{d}
\end{equation}
Where $\sigma$ is a hermitian operator satisfying $\tr(\sigma^2)=d$, $\tr(\sigma)=0$ and $c$ is a constant. We have
\begin{equation}
    \tr(W^2)=\tr(\frac{\mathbb{I}}{d^2})+\tr(\frac{c^2\sigma^2}{d^2})+2\tr(\frac{c\sigma}{d^2})=\frac{1+c^2}{d}.
\end{equation}
To show $\alpha^2=\frac{\tr(W)^2}{\tr(W^2)}\geq 1$, we are going to prove that $c^2\leq d-1$ by constructing a state:
\begin{equation}
    \rho_0=\frac{\mathbb{I}}{d}-\frac{\sigma}{\sqrt{d-1}d}.
\end{equation}
According to \cite{gurvits2002largest}, the set of separable states has a non-zero inner radius
\begin{equation}
    \forall \rho\in\mathcal{D}(\mathcal{H}_d), \tr(\rho^2)\leq \frac{1}{d-1}\to \rho \in \mathrm{SEP}.
\end{equation}
We can directly verify that
\begin{equation}
    \tr(\rho^2_0)=\frac{1}{d}+\frac{1}{d(d-1)}=\frac{1}{d-1},
\end{equation}
which means $\rho_0$ is not only a valid state but also separable. Since $W$ is an EW, $\tr(W\rho_0)\geq 0$.
\begin{equation}
\begin{split}
\tr(W\rho_0)&\geq 0\\
\frac{1}{d}-\frac{c}{d\sqrt{d-1}}&\geq 0\\
c^2&\leq d-1
\end{split}
\end{equation}
So we can conclude that for any valid entanglement witness $W$,
\begin{equation}
    \alpha^2=\frac{\tr(W)^2}{\tr(W^2)}=\frac{d}{1+c^2}\geq 1.
\end{equation}
\end{proof}

\begin{figure}[htbp]
    \centering
    \includegraphics[width=9cm]{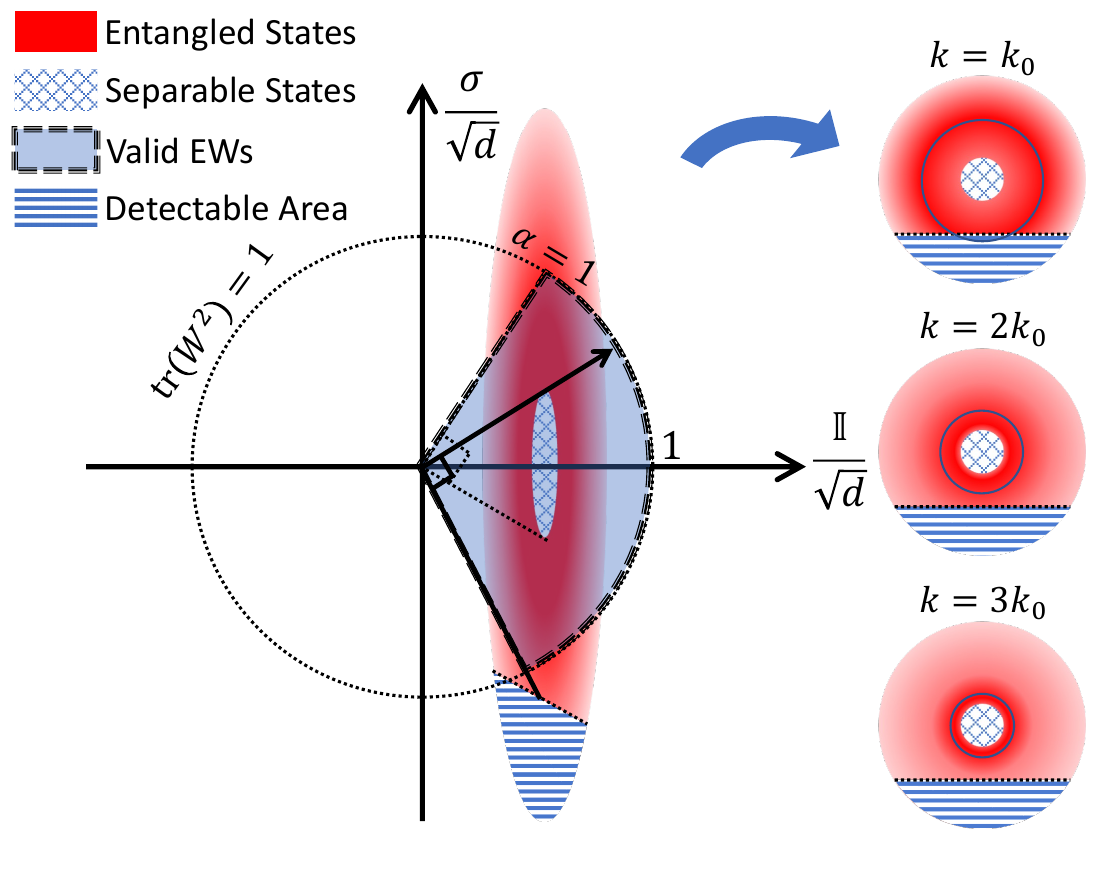}
    \caption{Graphical illustration of EW criteria. }
    \label{fig:geometry}
\end{figure}

We could also provide a graphical illustration of this lemma, which can help us understand the EW criteria. In Fig.~\ref{fig:geometry}, we use Pauli-Liouville representation to represent the density matrix and EWs as vectors in the operator space. The normalized identity $\frac{\mathbb{I}}{\sqrt{d}}$ is the $x$-axis, and the $y$-axis represents one of the other Pauli basis. The expectation of an observable can be calculated by the inner product of the state vector and the observable vector. Because of the trace condition $\tr(\rho)=1$, the density matrix lies in a hyperplane that is orthogonal to the $x$-axis. We use the solid and meshed area to represent the entangled and separable states. An EW can detect those states labeled by horizontal lines with obtuse angles with the EW. This observation and the fact that any state with a distance to the maximally mixed state less than a certain threshold is separable ensures that the angle between a valid EW and $y$-axis is larger than some constant. Quantitatively speaking, this tells us that $\alpha=\frac{\tr(W)}{\sqrt{\tr(W^2)}}$ has a minimum value. Without loss of generality, we assume all the EWs satisfy a normalization condition, $\tr(W^2)=\frac{1}{d}$. Thus all the EWs lie in a sphere centralized at the original point, represented by the dashed circle. Due to the $\alpha \geq 1$ ($\tr(W)\geq \frac{1}{\sqrt{d}}$) constraint, valid EWs are within dashed circular sector area. Given an EW, the states that can be detected lie in a fixed area in the space. When $k$ increases, the distribution of states, represented by the darkness of the color, will concentrate towards the maximally mixed state, making the ratio of detectable states decrease accordingly. The solid-line circles on the right represent the boundary of the typical set.

\subsection{Proof of Theorem 2}
\begin{theorem}[Detection Capability of EW Criteria]\label{theorem:singleEW}
The detection capability of an EW criterion with $W$ decays at least exponentially with the dimension of the environment
\begin{equation}
    \mathcal{C}_k(W)< 2e^{-(\sqrt{1+\alpha}-1)^2 k}\leq 2e^{-(3-2\sqrt{2}) k},
\end{equation}
where $\alpha=\frac{\tr(W)}{\sqrt{\tr(W^2)}}\geq 1$ is a witness-dependent factor.
\end{theorem}

\begin{proof}
We first generalize the definition of detection capability for EW to any observable $O$ with a positive trace and prove it with this generalized definition. Similarly, define
\begin{equation}
    \mathcal{C}_k(O)=\Pr_{\rho\sim\pi_{d,k}}\bqty{\tr(O\rho)<0}
\end{equation}
as the detection capability of an observable $O$.

Since $O$ can be decomposed as $O=U_O\Lambda_OU_O^\dagger$, where $U_O$ is unitary, and $\Lambda_O$ are the eigenvalues of $O$, we can equivalently rewrite
\begin{equation}
    \mathcal{C}_k(O)=\Pr_{\rho\sim\pi_{d,k}}\bqty{\tr(U_O\Lambda_OU_O^\dagger \rho)<0}=\Pr_{\rho\sim\pi_{d,k}}\bqty{\tr(\Lambda_OU_O^\dagger \rho U_O)<0}.
\end{equation} 
According to the definition of Haar measure, if $\rho$ follows the distribution of $\pi_{d,k}$, $U_O^\dagger\rho U_O$ also follows the distribution of $\pi_{d,k}$ as $U_O$ is a fix unitary \cite{Nechita2007}. Therefore, \begin{equation}
    \mathcal{C}_k(O)=\mathcal{C}_k(\Lambda_O)
\end{equation} only depends on the eigenvalues of $O$.

To analyze $\mathcal{C}_k(O)$, we need to write down the distribution of $\rho$ explicitly. According to the definition of $\pi_{d,k}$, $\rho\in\mathcal{D}(\mathcal{H})$ can be written as the reduced density matrix in a larger Hilbert space, $\rho=\tr_R(\ketbra{\Psi})$, where $\ket\Psi$ is a random state in $\mathcal{H}\otimes \mathcal{H}_R$. The distribution of $\ket\Psi$ can be generated by random Gaussian variables:
\begin{equation}
\ket\Psi=\sum_{i=1}^d\sum_{j=1}^k \frac{z_{i,j}}{\sqrt{\tr(ZZ^\dagger)}}\ket{\phi_i}\ket{\psi_j},
\end{equation}
where $z_{i,j}$ is the element of the random complex Gaussian matrix, $Z$ \cite{Nechita2007}; $\{\ket{\phi_i}\}$ and $\{\ket{\psi_j}\}$ form orthonormal bases for $\mathcal{H}$ and $\mathcal{H}_R$ respectively. Precisely speaking,
\begin{equation}\label{eq:distribution}
    x_{i,j}=\Re(z_{i,j})\sim N(0, 1), y_{i,j}=\Im(z_{i,j})\sim N(0, 1),
\end{equation}
are all standard Gaussian variables. Hence,
\begin{equation}
\rho=\tr_R(\ketbra{\Psi})=\sum_{i,j=1}^d \frac{\sum_{l=1}^k z_{i,l}z^*_{j,l}}{\tr(ZZ^\dagger)}\ket{\phi_i}\bra{\phi_j}.
\end{equation}
Therefore, the detection capability can be written as
\begin{equation}\label{eq:probGaussian}
\begin{split}
\Pr_{\rho\sim\pi_{d,k}}\bqty{\tr(\Lambda_O\rho)<0}=&\Pr_{x_{i,j},y_{i,j}\sim N(0, 1)}\bqty{\pqty{\frac{1}{\tr(ZZ^\dagger)}}\sum_{i=1}^{d} \sum_{j=1}^k\lambda_i z_{i,j}z_{i,j}^*<0}\\
=&\Pr_{x_{i,j},y_{i,j}\sim N(0, 1)}\bqty{\sum_{i=1}^{d} \sum_{j=1}^k\lambda_i(x_{i,j}^2+y_{i,j}^2)<0}.
\end{split}
\end{equation}
We label the positive and negative eigenvalues of $O$ as $a_1, \cdots, a_p$ and $-b_1\cdots -b_q$ with $a_i, b_j>0$ and $p+q\le d$. So we can rewrite Eq.~\eqref{eq:probGaussian} as
\begin{equation}
 \Pr_{x_{i,j},y_{i,j}\sim N(0, 1)}\bqty{\sum_{i=1}^{d} \sum_{j=1}^k\lambda_i(x_{i,j}^2+y_{i,j}^2)<0}=\Pr_{x_{i,j},y_{i,j}\sim N(0, 1)}\bqty{\sum_{j=1}^{2k}\pqty{\sum_{i=1}^p a_ix_{i,j}^2-\sum_{i=1}^q b_i y_{i,j}^2}<0},\label{Proof of theorem 1:step3}
\end{equation}
where the $x$ and $y$ in the left-hand and right-hand sides are not the same variables. We relabel them to make the representation clearer while keeping them independent variables. 

For simplicity, we define $\vb* u=(a_1,\dots,a_p,\dots,a_1,\dots,a_p)$ and $\vb* v=(b_1,\dots,b_q,\dots,b_1,\dots,b_q)$ which are the $2k$ replica of $\vb* a=(a_1,\dots,a_p)$ and $\vb* b=(b_1,\dots,b_q)$ respectively. Accordingly, 
\begin{equation}
\Pr_{x_{i,j},y_{i,j}\sim N(0, 1)}\bqty{\sum_{j=1}^{2k}\pqty{\sum_{i=1}^p a_ix_{i,j}^2-\sum_{i=1}^q b_i y_{i,j}^2}<0}=\Pr_{x_i,y_i\sim N(0,1)}\left(\sum_{i=1}^{2kp}u_ix_i^2-\sum_{i=1}^{2kq}v_iy_i^2<0\right).
\end{equation}
Using union bound, we can prove that for any real number $c$
\begin{equation}
\begin{split}
\Pr_{x_i,y_i\sim N(0,1)}\left(\sum_{i=1}^{2kp}u_ix_i^2-\sum_{i=1}^{2kq}v_iy_i^2<0\right)&\le\Pr_{x_i,y_i\sim N(0,1)}\left(\left(\sum_{i=1}^{2kp}u_ix_i^2\le c\right)\cup\left(\sum_{i=1}^{2kq}v_iy_i^2\ge c\right)\right)\\
&\le\Pr_{x_i\sim N(0,1)}\left(\sum_{i=1}^{2kp}u_ix_i^2\le c\right)+\Pr_{y_i\sim N(0,1)}\left(\sum_{i=1}^{2kq}v_iy_i^2\ge c\right).
\end{split}
\end{equation}
To bound this probability, we adopt the Laurent-Massart's lemma \cite{laurent2000adaptive}, which states that for non-negative vectors $\vb* u$ and $\vb* v$ and i.i.d. variables $\{x_i\sim N(0,1)\}$, the following two inequalities hold for all positive numbers $t_1$ and $t_2$:
\begin{equation}
\begin{split}
&\Pr_{x_i\sim N(0,1)}\left(\sum_i u_i x_i^2\leq \norm{\vb* u}_1-2\norm{\vb* u}_2\sqrt {t_2}\right)\leq e^{-t_2}\\
&\Pr_{y_i\sim N(0,1)}\left(\sum_i v_i y_i^2\geq \norm{\vb* v}_1+2\norm{\vb* v}_2\sqrt {t_1}+2\norm{\vb* v}_\infty t_1\right)\leq e^{-t_1}
\end{split}
\end{equation}
where $\norm{\vb* v}_1=\sum_i\abs{v_i}$ ,$\norm{\vb* v}_2=\sqrt{\sum_iv_i^2}$ and $\norm{\vb* v}_\infty=\max_i\abs{v_i}$. Hence, if 
\begin{equation}\label{eq:t1t2relation}
\begin{split}
&\norm{\vb* u}_1-2\norm{\vb* u}_2\sqrt{t_2}=2k\norm{\vb* a}_1-2\sqrt{2k}\norm{\vb* a}_2\sqrt{t_2}=c\\
&\norm{\vb* v}_1+2\norm{\vb* v}_2\sqrt{t_1}+2\norm{\vb* v}_\infty t_1=2k\norm{\vb* b}_1+2\sqrt{2k}\norm{\vb* b}_2\sqrt{t_1}+2\norm{\vb* b}_\infty t_1=c
\end{split}
\end{equation}
hold, then the probability can be upper bounded by
\begin{equation}\label{eq:upperbound12}
\Pr_{x_i,y_i\sim N(0,1)}\left(\sum_{i=1}^{2kp}u_ix_i^2-\sum_{i=1}^{2kq}v_iy_i^2<0\right)\le e^{-t_1}+e^{-t_2}.
\end{equation}

To find a $c$ that gives the tightest bound, one should notice that according to Eq.~\eqref{eq:upperbound12}, the upper bound is determined by the minimal one of $t_1$ and $t_2$. Besides, Eq.~\eqref{eq:t1t2relation} tells us that the values of $t_1$ and $t_2$ are inversely related. Therefore, the tightest upper bound is reached when $t_1=t_2=t$, which gives the exact value of $t$:
\begin{equation}
 \sqrt {\frac{t}{2k}}=\frac{-(\norm{\vb* a}_2+\norm{\vb* b}_2)+\sqrt{{(\norm{\vb* a}_2+\norm{\vb* b}_2)}^2+2\norm{\vb* b}_\infty \tr(O)}}{2\norm{\vb* b}_\infty}.
\end{equation}
To further simplify this equation, let $\alpha=\frac{\tr(O)}{\sqrt{\tr(O^2)}}$, then we have
\begin{equation}
\begin{split}
 \sqrt {\frac{t}{2k}}&=\frac{-(\norm{\vb* a}_2+\norm{\vb* b}_2)+\sqrt{{(\norm{\vb* a}_2+\norm{\vb* b}_2)}^2+2\norm{\vb* b}_\infty \tr(O)}}{2\norm{\vb* b}_\infty}\\&\geq
 \frac{-\sqrt{2\tr(O^2)}+\sqrt{2\tr(O^2)+2\norm{\vb* b}_\infty \tr(O)}}{2\norm{\vb* b}_{\infty}}\\&=
 \frac{-\alpha^{-1}\tr(O)+\sqrt{\alpha^{-2}{\tr(O)}^2+\norm{\vb* b}_\infty \tr(O)}}{\sqrt{2}\norm{\vb* b}_\infty}.
 \end{split}
\end{equation}
The first inequality uses the fact that $f(x)=-x+\sqrt{1+x^2}$ is monotone and $\tr(O^2)=\norm{\vb* a}_2^2+\norm{\vb* b}_2^2\geq \frac{{(\norm{\vb* a}_2+\norm{\vb* b}_2)}^2}{2}$. Define
\begin{equation}
    x=\frac{\alpha\norm{\vb* b}_\infty}{\tr(O)}=\frac{\norm{\vb* b}_\infty}{\sqrt{\tr(O^2)}}\geq0.
\end{equation} 
By definition, $\norm{\vb* b}_\infty=\max_i|b_i|<\sqrt{\tr(O^2)}$, it is easy to prove that $0\le x\le \sqrt{1-\frac{1}{d}}< 1$. Therefore,
\begin{equation}
    \sqrt{\frac{t}{k}}= \frac{-1+\sqrt{1+\alpha x}}{x}>\sqrt{1+\alpha}-1,
\end{equation}
where we use the fact that the function $\frac{-1+\sqrt{1+\alpha x}}{x}$ is monotonically decreasing with $x$. Combined with Eq.~\eqref{eq:upperbound12}, we have
\begin{equation}
 \mathcal{C}_k(O) < 2e^{-{(\sqrt{1+\alpha}-1)}^2k}.
\end{equation}
\end{proof}

\section{Detection Capability Upper Bound of Parameterized EW Criteria}\label{sec:Para EW}
\subsection{Proof of Theorem 3}
\begin{theorem}[Detection Capability of Parameterized EW Criteria]
For any parameterized EW represented by a normalized $l$-Lipschitz map $\mathcal{M}$ satisfying
\begin{equation}
\forall \vb*\theta,\vb*\theta'\in \Theta:\norm{\mathcal{M}(\vb*\theta)-\mathcal{M}(\vb*\theta')}_F\leq l\norm{\vb*\theta-\vb*\theta'}_2,
\end{equation}
The detection capability decays at least exponentially with  $k$ after $k$ exceeds a certain threshold
\begin{equation}
    \mathcal{C}^p_k(\mathcal{M})\leq 2e^{C_1-C_2k},
\end{equation}
where $C_1=M\ln \frac{2\sqrt{M}ld}{\epsilon}$, $M$ is the number of real parameters in $\mathcal{M}$, $C_2=(\sqrt{1+\alpha_{\min}-\epsilon}-1)^2$ where $\alpha_{\min}=\min_{\vb*\theta} \frac{\tr[\mathcal{M}(\vb*\theta)]}{\sqrt{\tr[\mathcal{M}(\vb*\theta)^2]}}=\min_{\vb*\theta} \tr[\mathcal{M}(\vb*\theta)]\geq 1$, and $0< \epsilon<1$ is an arbitrary number. By choosing $\epsilon=0.5$, we have the original theorem in the main text.
\label{theorem:paraEWsinapp}
\end{theorem}
We prove this theorem using a coarse-graining method. The proof sketch is shown in Fig.~\ref{fig:proofoftheo2}.
\begin{figure}[htbp]
\centering
\includegraphics[width=7cm]{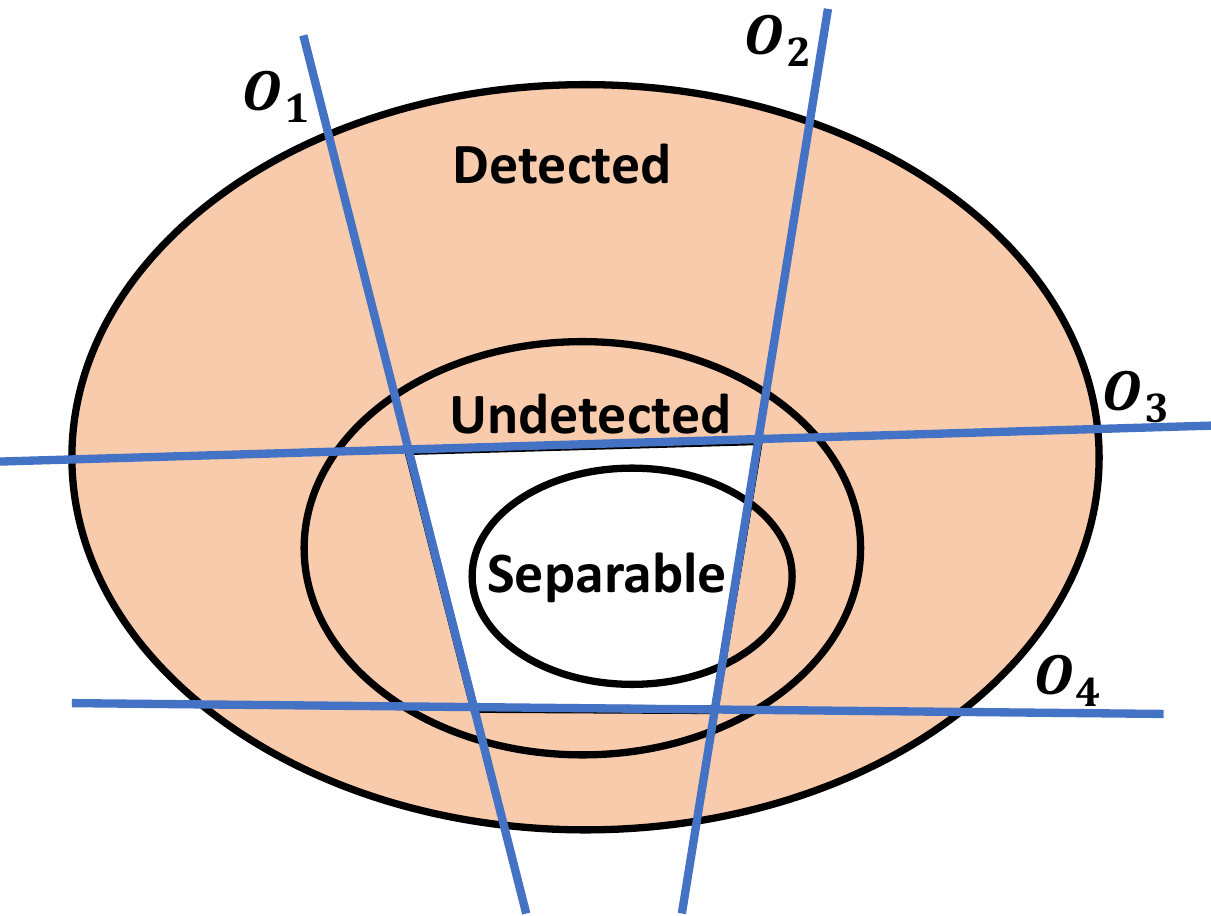}
\caption{Proof Sketch of Theorem \ref{theorem:paraEWsinapp}. We prove that if $\mathcal{M}$ is $l$-Lipschitz, the set of detected states can be covered by a finite set of different observables, like the colored area in this figure. Each observable is close to but not necessarily an EW. Therefore, by bounding the total number of observables and the difference between these observables and valid EWs, the volume of the detected set can be upper bounded by Theorem \ref{theorem:singleEW}.}
\label{fig:proofoftheo2}
\end{figure}

\begin{proof}
A parameterized EW $\mathcal{M}$ is a map that maps $M$ real parameters to a continuous set of  EWs:
\begin{equation}
    \forall \vb* \theta\in\Theta\subset[-1,1]^M,\rho\in \mathrm{SEP}:\tr(\rho \mathcal{M}(\vb* \theta))\geq 0.
\end{equation}
We are going to bound the detection capability of a parameterized EW 
\begin{equation}
\mathcal{C}^p_k(\mathcal{M})=\Pr_{\rho\sim\pi_{d,k}}\bqty{\exists \vb*\theta\in\Theta: \tr(\rho \mathcal{M}(\vb* \theta))<0}
\end{equation}
by constructing a finite set of observables $\mathcal{O}=\Bqty{O_i|i=1,\cdots, N}$ (not necessarily EWs), such that all the entangled states $\rho$ that can be detected by $\mathcal{M}$ can also be detected by $\mathcal{O}$,
\begin{equation}
        \forall \rho : \exists \vb*\theta\in\Theta, \tr\left(\rho \mathcal{M}(\vb* \theta)\right)<0\to \exists O_i\in \mathcal{O},\tr\pqty{\rho O_i}<0.
\end{equation}
Once we find the observable set $\mathcal{O}$, the detection capability of $\mathcal{M}$ is bounded by the detection capability of $\mathcal{O}$,
\begin{equation}
    \mathcal{C}^p_k(\mathcal{M})\leq \mathcal{C}_k(\mathcal{O}) = \Pr_{\rho\sim\pi_{d,k}}\left[\exists O\in\mathcal{O},\tr(O\rho)<0\right].
\end{equation}

Firstly, we coarse-grain the parameter space, define $\Theta^*=\Bqty{\vb* \theta_i\in\Theta,i=1,\cdots, N}$, such that
\begin{equation}\label{eq:coarsegraincondition}
    \forall \vb*\theta \in \Theta, \exists \vb* \theta^*\in \Theta^*:\norm{\vb* \theta-\vb* \theta^*}_2\leq \delta.
\end{equation}
Since $\mathcal{M}$ is $l$-Lipschitz, we have
\begin{equation}
    \forall\vb*\theta \in \Theta, \exists \vb* \theta^*\in\Theta^*:\norm{\mathcal{M}(\vb* \theta)-\mathcal{M}(\vb* \theta^*)}_F\leq l\delta,
\end{equation}
which means that
\begin{equation}
    \forall \vb*\theta \in \Theta, \exists \vb* \theta^*\in\Theta^*: \mathcal{M}(\vb* \theta)-(\mathcal{M} (\vb* \theta^*)-l\delta  \mathbb{I}) \geq 0.
\end{equation}
Hence, for any state $\rho$ satisfying $\tr\left(\rho\mathcal{M}(\vb* \theta)\right)<0$, it also holds that
\begin{equation}
    \exists\vb*\theta^*\in\Theta^*:\tr\left[\left(\mathcal{M}(\vb* \theta^*)-l\delta \mathbb{I}\right)\rho\right]<0.
\end{equation}
Therefore, we can choose $\mathcal{O}$ to be $\mathcal{O}=\Bqty{\mathcal{M}(\vb* \theta_i)-l\delta\mathbb{I},\vb* \theta_i\in\Theta^*}$, whose detection capability can also be bounded using Theorem \ref{theorem:singleEW}.

To bound $\mathcal{C}_k(\mathcal{O})$, we need to figure out two problems: what is the detection capability of a single $O_i$ in $\mathcal{O}$ and what is the number of elements in $\mathcal{O}$. According to Theorem \ref{theorem:singleEW}, the key quantity to bound $\mathcal{C}_k\left(\mathcal{M}(\vb* \theta^*)-l\delta\mathbb{I}\right)$ is
\begin{equation}
\begin{split}
    \alpha^2=\frac{\tr(\mathcal{M}(\vb* \theta^*)-l\delta  \mathbb{I})^2}{\tr((\mathcal{M}(\vb* \theta^*)-l\delta  \mathbb{I})^2)}=\frac{(\alpha^*-l\delta d)^2}{1-2l\delta\alpha^*+(l\delta)^2d},\label{proof of theorem 2:step1}
\end{split}
\end{equation}
where $\alpha^*=\frac{\tr(\mathcal{M}(\vb* \theta^*))}{\sqrt{\tr(\mathcal{M}(\vb* \theta^*)^2)}}=\tr(\mathcal{M}(\vb* \theta^*)) \geq 1$. It can also be directly verified by norm inequality, $\alpha^*\leq \sqrt{d}$. Define $0< l\delta d= \epsilon< 1$, we have
\begin{equation}
    1-2l\delta\alpha^*+(l\delta)^2d=1-2\frac{\epsilon\alpha^*}{d}+\frac{\epsilon^2}{d}>0
\end{equation}
and
\begin{equation}
    1-2\frac{\epsilon\alpha^*}{d}+\frac{\epsilon^2}{d}\leq 1
\end{equation}
Combine these inequalities with Eq.~\eqref{proof of theorem 2:step1}, and we get $\alpha\geq \alpha^*-\epsilon$. Thus the detection capability of a single observable in $\mathcal{O}$ can be bounded by
\begin{equation}\label{proof of theorem 2:step2}
\mathcal{C}_k\left(\mathcal{M}(\vb* \theta^*)-l\delta\mathbb{I}\right)< 2e^{-(\sqrt{1+\alpha_{\mathrm{min}}-\epsilon}-1)^2k},
\end{equation}
where $\alpha_{\min}=\min_{\vb*\theta\in \Theta} \frac{\tr(\mathcal{M}(\vb* \theta))}{\sqrt{\tr(\mathcal{M}(\vb* \theta)^2)}}\geq 1$.

To find the number of elements in $\mathcal{O}$, we can divide the parameter space into small cubes with side length $\frac{\delta}{\sqrt{M}}$. In each cube, there exists a $\vb*\theta_i$, such that for all the $\vb*\theta$ contained in this cube, $\norm{\vb*\theta-\vb*\theta_i}_2\le\sqrt{M\left(\frac{\delta}{\sqrt{M}}\right)^2}=\delta$, which fulfills the condition of Eq.~\eqref{eq:coarsegraincondition}. As the volume of parameter space is upper bounded by $2^M$, the number of cubes, which is also the upper bound of the number of elements in $\mathcal{O}$, is 
\begin{equation}\label{eq:numberofobs}
\abs{\mathcal{O}}=\left(\frac{2\sqrt{M}}{\delta}\right)^M=\left(\frac{2\sqrt{M}ld}{\epsilon}\right)^M.
\end{equation}
Combining Eq.~\eqref{proof of theorem 2:step2} and Eq.~\eqref{eq:numberofobs}, we can finish the proof by
\begin{equation}
\mathcal{C}^p_k(\mathcal{M})\le\mathcal{C}_k(\mathcal{O})<2 e^{M\ln \frac{2\sqrt{M}ld}{\epsilon}-(\sqrt{1+\alpha_{min}-\epsilon}-1)^2k}.
\end{equation}
\end{proof}

\subsection{Examples: Positive Map and Faithful Entanglement Criteria}

A bipartite state $\rho\in\mathcal{D}(\mathcal{H}_A\otimes\mathcal{H}_B)$ can be detected by a positive map $\mathcal{N}$ if and only if there exists a parameterized EW $\mathcal{M}_{\mathcal{N}}\pqty{\vb*\theta}=\frac{\mathcal{N}_A\otimes\mathbb{I}_B(\ketbra{\phi(\vb*\theta)})}{\norm{\mathcal{N}_A\otimes\mathbb{I}_B(\ketbra{\phi(\vb*\theta)})}_F}$ to detect it. This is equivalent to
\begin{equation}
    \exists\vb*\theta\in S^{2d-1}: \tr\bqty{\mathcal{N}_A\otimes\mathbb{I}_B\left(\ketbra{\phi(\vb*\theta)}\right)\rho}<0,
\end{equation}
where $S^{2d-1}$ is the unit sphere in the $2d$-dimensional parameter space, and $\bra{j}\ket{\phi(\vb* \theta)}=\theta_{2j}+i\theta_{2j+1}$. Hence, substituting $M$ with $2d$, we have:
\begin{corollary}[Detection Capability of Positive Maps]\label{coro:PNCP}
A normalized $l$-Lipschitz positive map $\mathcal{N}$ has detection capability:
\begin{equation}
\mathcal{C}^p_k(\mathcal{M}_{\mathcal{N}})< 2e^{C_1-C_2k}
\end{equation}
where $C_1=2d\ln \bqty{2^{2.5}d^{1.5}l}$, $C_2=(\sqrt{0.5+\alpha_{\min}}-1)^2$, $\alpha_{\mathrm{min}}=\min_{\vb*\theta}\frac{\tr[\mathcal{N}_A\otimes  \mathbb{I}_B(\ketbra{\phi(\vb*\theta)})]}{\sqrt{\tr[\mathcal{N}_A\otimes  \mathbb{I}_B(\ketbra{\phi(\vb*\theta)})^2]}}$.
\end{corollary}
\begin{proof}
It follows directly from Theorem~\ref{theorem:paraEWsinapp} by choosing $\epsilon=0.5$ and $M=2d$.
\end{proof}
Also take the PPT criterion as an example, where $\mathcal{M}_{\mathrm{PPT}}(\vb*\theta)=\ketbra{\phi(\vb*\theta)}^{{T_A}}$. It can be easily proved that the partial transposition map is $\sqrt{2}$-Lipschitz and $\alpha_{\mathrm{min}}=1$. First, we give the relationship between the $F$-norm of the density matrix representation and the $2$-norm of the real-valued vector representation. 
\begin{equation}
\begin{split}
    \norm{{\vb*\theta}-{\vb*\theta'}}_2^2&=\norm{\ket{\phi(\vb*\theta)}-\ket{\phi(\vb*\theta')}}_2^2\\&=2-2\Re(\braket{\phi(\vb*\theta)}{\phi(\vb*\theta')})\\&\geq 2-2\abs{\bra{\phi(\vb*\theta)}\ket{\phi(\vb*\theta')}},
\end{split}
\end{equation}
\begin{equation}
\begin{split}
    \norm{\ketbra{\phi(\vb*\theta)}-\ketbra{\phi(\vb*\theta')}}_F^2&=2-2\abs{\bra{\phi(\vb*\theta)}\ket{\phi(\vb*\theta')}}^2\\&=(2-2\abs{\bra{\phi(\vb*\theta)}\ket{\phi(\vb*\theta')}})(1+\abs{\bra{\phi(\vb*\theta)}\ket{\phi(\vb*\theta')}})\\&\leq 2\norm{{\vb*\theta}-{\vb*\theta'}}_2^2
\end{split}
\end{equation}
So the map $\mathcal{M}(\vb*\theta)=\ketbra{\phi(\vb*\theta)}$ is $\sqrt{2}$-Lipschitz. 

For partial transposition map, $\mathcal{M}(\vb*\theta)=\ketbra{\phi(\vb*\theta)}^{{T_B}}$, it is normalized by itself,
\begin{equation}
    \norm{\mathcal{M}(\vb*\theta)}=\norm{\ketbra{\phi(\vb*\theta)}^{{T_B}}}_F=\norm{\ketbra{\phi(\vb*\theta)}}_F=1
\end{equation}
and
\begin{equation}
\begin{split}
        \norm{\mathcal{M}(\vb*\theta)-\mathcal{M}(\vb*\theta')}&=\norm{\ketbra{\phi(\vb*\theta)}^{T_B}-\ketbra{\phi(\vb*\theta)}^{T_B}}_F\\&=\norm{\ketbra{\phi(\vb*\theta)}-\ketbra{\phi(\vb*\theta)}}_F\\&\leq \sqrt{2}\norm{\vb*\theta-\vb*\theta'}_2.
\end{split}
\end{equation}
Therefore partial transposition map is $\sqrt{2}$-Lipschitz. Corollary \ref{coro:PNCP} shows that for $k=\Omega(d\ln d)$, the PPT criterion can hardly detect any entanglement, which meets the former results \cite{nidari2007witness,bhosale2012entanglement,shapourian2021diagram}. 

For faithful EW, the parameterized EW can be defined as
\begin{equation}
\mathcal{M}(\vb*\theta)=\pqty{\sqrt{2-\frac{2}{\sqrt{d}}}}^{-1}\pqty{\frac{\mathbb{I}}{\sqrt{d}}-\ketbra{\phi(\vb*\theta)}}
\end{equation}
where $\pqty{\sqrt{2-\frac{2}{\sqrt{d}}}}^{-1}\leq 1$ is a factor to ensure $ \mathcal{M}(\vb*\theta)$ is normalized, then
\begin{equation}
\begin{split}
        \norm{\mathcal{M}(\vb*\theta)-\mathcal{M}(\vb*\theta')}_F&=\pqty{\sqrt{2-\frac{2}{\sqrt{d}}}}^{-1}\norm{\ketbra{\phi(\vb*\theta)}-\ketbra{\phi(\vb*\theta')}}_F\\&\leq \sqrt{2}\norm{\vb*\theta-\vb*\theta'}_2.
\end{split}
\end{equation}
So faithful map is also $\sqrt{2}$-Lipschitz with $2d$ real parameters. Combined with the fact that $\alpha_{\min}=\sqrt{\frac{d-\sqrt{d}}{2}}$, we have
\begin{corollary}[Ratio of Faithful Entanglement States]
The set of faithful entangled states has an exponentially small ratio in the state space:
\begin{equation}
        \Pr_{\rho\sim \pi_{d,k}}[\rho\in \mathrm{FE}]=\mathcal{C}^p_k\pqty{\mathcal{M}_{\mathrm{faithful}}}< 2e^{C_1-C_2k}
\end{equation}
where $\mathrm{FE}$ is the set of all faithful entangled states and $C_1=3d\ln 4d$, $C_2=\pqty{\sqrt{0.5+\sqrt{\frac{d-\sqrt{d}}{2}}}-1}^2\approx \sqrt{\frac{d}{2}}$.
\end{corollary}

\section{Detection Capability Upper Bound of Single-copy Criteria}\label{sec:singlecopy}
\subsection{Proof of Theorem 4}

\begin{theorem}[Detection Capability of Single-Copy Criteria]
\label{theorem:singlecopy_inapp}
Any single-copy entanglement criterion $\mathcal{O}$ with $M-1$ observables has detection capability:
\begin{equation}
    \mathcal{C}^s_k(\mathcal{O})\leq 2e^{C_1-C_2k}
\end{equation}
Where $C_1=M\ln \frac{2\sqrt{M}d}{\epsilon}$, $C_2=(\sqrt{2-\epsilon}-1)^2$. $0<\epsilon<1$ is an arbitrary number. By choosing $\epsilon=0.5$, we have the original theorem in the main text.
\end{theorem}
\begin{proof}
Without loss of generality, we add $O_M=\frac{\mathbb{I}}{\sqrt{d}}$ to the set, so $\mathcal{O}$ has $M$ observables now. We further assume $\mathcal{O}$ is mutually orthonormal in the operator space $\tr(O_iO_j)=\delta_{ij}$. If this condition is not satisfied, we can normalize and orthogonalize the operator set without changing the feasible region. 
Given the observable set $\mathcal{O}$ with $M$ observables and the measurement result
\begin{equation}
    r_{\rho,i}=\tr(O_i\rho), i=1\cdots M,
\end{equation}
the quantum state is restricted in the feasible region defined as
\begin{equation}
    F_\mathcal{O}(\rho)=\Bqty{\sigma\in\mathcal{D}(\mathcal{H}_d)|\tr(O_i\sigma)=r_{\rho,i},i=1\cdots M}.
\end{equation}
If the feasible region is disjoint with SEP, then the entanglement is successfully detected by $\mathcal{O}$. Therefore, the detection capability of $\mathcal{O}$ is defined as
\begin{equation} 
        \mathcal{C}^s_k(\mathcal{O})=\Pr_{\rho\sim\pi_{d,k}}\bqty{F_{\mathcal{O}}(\rho)\cap \mathrm{SEP}=\varnothing}.
\end{equation}

To benefit our proof, we extend the definition of $F_{\mathcal{O}}(\rho)$ from density states to Hermitian matrices, define
\begin{equation}
    F_{\mathcal{O}}'(\rho)=\Bqty{\sigma|\tr(O_i\sigma)=r_{\rho,i},\sigma^\dagger=\sigma}.
\end{equation}
It is easy to prove that $F_{\mathcal{O}}(\rho)\cap \mathrm{SEP}=\varnothing$ if and only if $F_{\mathcal{O}}'(\rho)\cap \mathrm{SEP}=\varnothing$ as SEP is only in the density matrices set. By definition, SEP and $F_{\mathcal{O}}(\rho)$ are all convex sets. Hence, from the hyperplane separation theorem, we can find Hermitian operators $W$ that separate SEP and $F_{\mathcal{O}}'$
\begin{equation}
    \exists W: \tr(W\sigma)< 0, \forall \sigma\in F_{\mathcal{O}}'(\rho)\text{ and }\tr(W\sigma')\geq 0, \forall \sigma'\in \mathrm{SEP}\label{proof of theorem3:step1},
\end{equation}
which is also an EW separate SEP and $F_{\mathcal{O}}(\rho)$. It can be proved that $W$ must have the form
\begin{equation}
    W=\sum_{i=1}^{M} \theta_i O_i.
\end{equation}
If not, suppose $W=\sum_i \theta_i O_i+\tilde{O}$, where $\tilde{O}\neq 0$ is orthogonal to each $O_i\in\mathcal{O}$, Then for any $\sigma\in F_{\mathcal{O}}'(\rho)$, $\sigma+C\tilde{O}\in F_{\mathcal{O}}'(\rho)$ where $C$ is an arbitrary real number. In this scenario, $\tr\left((\sigma+C\tilde{O})W\right)=C\tr(\tilde{O}^2)$ can be arbitrary large, which contradicts the requirement \eqref{proof of theorem3:step1}.

Accordingly, the entangled states that can be detected by $\mathcal{O}$ can also be detected by the following parameterized EW
\begin{equation}
    \mathcal{M}(\vb*\theta)=\sum_{i=1}^{M} \theta_i O_i
\end{equation}
where $\vb*\theta\in\Theta$ and $\Theta$ is constituted by all $\vb*\theta$ such that $\sum_i\theta_i^2=1$ and makes $\mathcal{M}(\vb*\theta)$ a valid EW. Therefore, the detection capability of the single-copy criteria $\mathcal{O}$ is bounded by the detection capability of the parameterized EW $\mathcal{M}$
\begin{equation}
    \mathcal{C}^s_k(\mathcal{O})\leq \mathcal{C}^p_k(\mathcal{M}).
\end{equation}
Such parameterized EW with $M$ parameters is normalized, and $1$-Lipschitz
\begin{equation}
\norm{W(\vb*\theta)-W(\vb*\theta)'}^2_F=\sum_i(\theta_i-\theta_i')^2\tr(O_i^2)=\norm{\vb*\theta-\vb*\theta'}^2_2.
\end{equation}
By directly applying Theorem \ref{theorem:paraEWsinapp}, we have
\begin{equation}
\mathcal{C}^s_k(\mathcal{O})< 2e^{M\ln\frac{2\sqrt{M}d}{\epsilon}-(\sqrt{2-\epsilon}-1)^2}
\end{equation}
Where $0<\epsilon<1$ is an arbitrary number. 
\end{proof}

\subsection{Adaptive Single-Copy Measurement}
The most general method to detect entanglement may take advantage of adaptive measurements. After the previous $j-1$ measurement, One can determine $O_j$ as a function of previous measurement results. Here we consider a case where each measurement or query gives $1$ bit of information.
\begin{definition}[Measurement with $1$ Bit Information]
The measurement can be viewed as a quantum oracle, given an observable $O$, the oracle will output $\mathrm{sign}(\tr(O\rho))$, more specifically, $+1$ if $\tr(O\rho)\geq 0$ and $-1$ if $\tr(O\rho)< 0$. 
\end{definition}
To determine whether $\tr(O\rho)\geq c$, one may simply replace $O$ by $O-c\mathbb{I}$. To determine any observable up to $\epsilon$ precision, one may use a binary search method with $O(\ln \frac{1}{\epsilon})$ queries. Next, we define the most general adaptive single-copy measurement where observables may depend on previous results. Formally, we define:
\begin{definition}[Adaptive Single-Copy Protocols]
An adaptive single-copy entanglement detection protocol with finite precision contains a program $\mathcal{P}$ that can generate an observable based on the previous results. More specifically, after the previous $j-1$ measurement, one get the measurement results $(k_1...k_{j-1})\in\{-1,+1\}^{j-1}$. Based on these result, the program can generate $O_j=f_j(k_i...k_{j-1})$. After $M$ iterations, one gets the following equations:
\begin{equation}
\mathrm{sign}(\tr(O_i\rho))=k_i\in\{-1,+1\},\forall i=1,\cdots, M
\end{equation}
We can still define the feasible set
\begin{equation}
    F_\mathcal{P}(\rho)=F_\mathcal{K}(k_1,\cdots, k_M)=\Bqty{\sigma\in\mathcal{D}(\mathcal{H}_d)|\mathrm{sign}(\tr(O_i\sigma))=k_i,i=1,\cdots,M}
\end{equation}
And the detection capability is similarly defined as
\begin{equation} 
    \mathcal{C}_k^s(\mathcal{P})=\Pr_{\rho\sim\pi_{d,k}}\bqty{F_{\mathcal{P}}(\rho)\cap \mathrm{SEP}=\varnothing}
\end{equation}
\end{definition}
Use the measurement outcome $\vb* k=(k_1,\cdots, k_M)\in\{-1,1\}^{M}$ to rewrite the previous definition:
\begin{equation}
    F_{\mathcal{P}}(\rho)\cap \mathrm{SEP}=\varnothing\Longleftrightarrow \exists \vb* k: \rho\in F_\mathcal{K}(\vb* k) \text{ and } F_{\mathcal{K}}(\vb* k)\cap \mathrm{SEP}=\varnothing
\end{equation}
So that
\begin{equation}
\begin{split}
        \mathcal{C}_k^s(\mathcal{P})&=\Pr_{\rho\sim\pi_{d,k}}\bqty{\exists \vb* k: \rho\in F_\mathcal{K}(\vb* k) \text{ and } F_{\mathcal{K}}(\vb* k)\cap \mathrm{SEP}=\varnothing}\\&=\sum_{\vb* k} \Pr_{\rho\sim\pi_{d,k}}\bqty{\rho\in F_\mathcal{K}(\vb* k) \text{ and } F_{\mathcal{K}}(\vb* k)\cap \mathrm{SEP}=\varnothing}
\end{split}
\end{equation}
Notice that for any $\vb* k$, $F_\mathcal{K}(\vb* k)$ is a convex set. So if $F_{\mathcal{K}}(\vb* k)\cap \mathrm{SEP}=\varnothing$, then by the hyperplane separation theorem, there exists an EW $W$ s.t. 
\begin{equation}
    \tr(W\sigma)< 0, \forall \sigma\in F_\mathcal{K}(\vb* k)\text{ and }\tr(W\sigma')\geq 0, \forall \sigma'\in \mathrm{SEP}
\end{equation}
According to Theorem \ref{theorem:singleEW}, each term in the summation is bounded by $2e^{-(3-2\sqrt{2})k}$. And there are a total of $2^M$ different terms in the summation,
\begin{equation}
 \mathcal{C}_k^s(\mathcal{P})\leq 2^{M+1}e^{-(3-2\sqrt{2})k}=2e^{M\ln2-(3-2\sqrt{2})k}.
\end{equation}
So the detection capability of any adaptive single-copy method also suffers from exponential decay.

\subsection{Details of Figure 3}
In Fig.~3 of the main text, we use four entanglement criteria to demonstrate our conclusion of the single-copy criteria. We explicitly list them here. Suppose the state $\rho_{AB}$ we consider is bipartite with subsystems $A$ and $B$.
\begin{enumerate}
    \item Purity \cite{GUHNE2009detection}:
    \begin{equation}
    \forall \rho_{AB}\in \mathrm{SEP}:\tr(\rho_{AB}^2)\le \tr(\rho_A^2).
    \end{equation}
    \item Fisher Information \cite{Zhang_2020}:
    \begin{equation}
       \forall \rho_{AB}\in \mathrm{SEP}: F(\rho, A \otimes I+I \otimes B) \leq \Delta(A \otimes I-I \otimes B)_{\rho}^{2},
    \end{equation}
where
\begin{equation}
F(\rho, A)=\sum_{k, l} \frac{(\lambda_k-\lambda_l)^2}{2\left(\lambda_{k}+\lambda_{l}\right)}\abs{\mel{k}{A}{l}}^2,
\end{equation}

\begin{equation}
    \rho=\sum_k \lambda_k\ketbra{k},
\end{equation}
and
\begin{equation}
    \Delta(A)_{\rho}^{2}=\expval{A^2}_\rho-\expval{A}^2_\rho.
\end{equation}
Since this criterion holds for any observable $A$ and $B$, we randomly choose $10$ different $A$s and $B$s to build a series of criteria. If any of them is violated, the state is classified as entangled.
    \item $M_4$ \cite{liu2022detecting}:
    \begin{equation}
       \forall \rho_{AB}\in \mathrm{SEP}: E_{4}(\rho_{AB}-\rho_A\otimes\rho_B)\le \sqrt{(1-\tr(\rho_A^2))(1-\tr(\rho_B^2))}.
    \end{equation}
        where $E_{4}(\rho)=\sqrt{\frac{q(q M_2+U)}{q+1}}+\sqrt{\frac{M_2-U}{q+1}}$,  $q=\lfloor\frac{M_2^2}{M_4}\rfloor$, $U=\sqrt{q(q+1)M_4-qM_2^2}$, $M_2=\tr(\rho_{AB}^2)$, $M_4=\tr[(\mathbb{S}_A^{(1,2)}\otimes\mathbb{S}_A^{(3,4)}\otimes\mathbb{S}_B^{(2,3)}\otimes\mathbb{S}_B^{(4,1)})\rho_{AB}^{\otimes 4}]$, and $\mathbb{S}_A^{(i,j)}$ is the SWAP operator acting on the $i$-th and $j$-th copies of subsystem $A$.
    \item $D_{3,\mathrm{opt}}$ \cite{yu2021optimal,neven2021symmetry}:
    \begin{equation}
        \forall \rho_{AB}\in \mathrm{SEP}: \beta x^3+(1-\beta x)^3\le\tr\left((\rho_{AB}^{T_B})^3\right),
    \end{equation}
    where $\beta=\lfloor\frac{1}{\tr(\rho_{AB}^2)}\rfloor$ and $x=\frac{\beta+\sqrt{\beta\left((\beta+1)\tr(\rho_{AB}^2)-1\right)}}{\beta(\beta+1)}$.
    \end{enumerate}

\subsection{More Numerical Experiments}
\subsubsection{Relationship between the threshold $k_{th}$ and $d$}
In Fig.~\ref{smfig:3criteria}, we show the detection capability of purity, $M_4$, and $D_{3,\mathrm{opt}}$ criteria. All the curves have two regimes: constant and exponential decay with $k$. Denote the turning point between these two regimes to be $k_{th}$, beyond which the criterion becomes ineffective. For different dimensions $d$, it is interesting to study the threshold $k_{th}$ for different criteria. From the figure, the thresholds $k_{th}$ for Purity, $M_4$, and $D_{3,\mathrm{opt}}$ are approximately $\sqrt{d}$, $0.6d$, and $d$, respectively. These polynomial relations show that an exponential number of observables are needed to verify these criteria with only single-copy observables. In fact, all three criteria require the number of observables larger than $\Omega(k_{th}/\ln k_{th})$ with the best-known randomized measurements. This is consistent with Theorem \ref{theorem:singlecopy_inapp}.

\begin{figure}[ht]
 \centering
 \includegraphics[width=16cm]{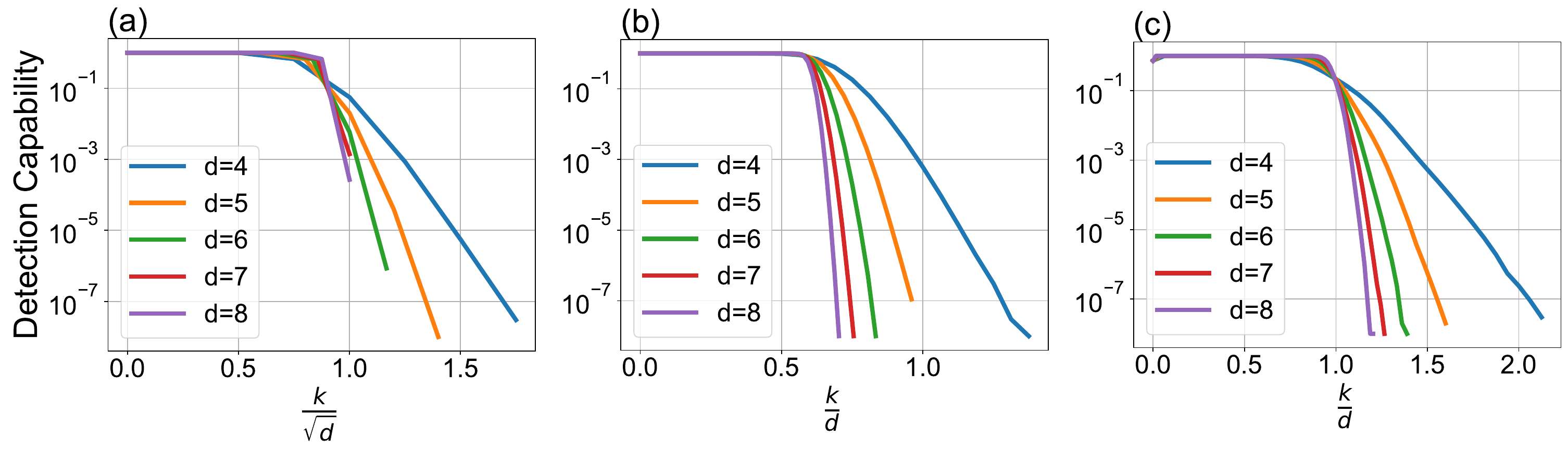}
\caption{The three figures represents the detection capability of purity, $M_4$ and $D_{3,\mathrm{opt}}$ from (a) to (c) respectively. In the exponential decaying period, each line represents $d=4,5,6,7,8$ from right to left. The $y$ axes are the detection capability, and $x$ axes are $k$ divided by factors depending on $d$. The factors are chosen so that the relation between $k_{th}$ and $d$ can be easily observed. Each point is generated through $10^8$ independent experiments.}
\label{smfig:3criteria}
\end{figure}

\subsubsection{Numerical experiments on random thermal states}
In the proof of the theorems, we assume distribution $\pi_{d,k}$. Obviously, the results cannot hold for all distributions. For example, if the states only distribute around a particular maximally entangled state, we can easily design an effective EW to witness these states. In this case, we already assume lots of prior information about the states. Without such strong prior information, the states are more evenly distributed over the state space. Then, if the state distribution is approximately symmetric around the maximally mixed state, the theorems should also hold. Here, we present another typical state distribution as an example and leave detailed studies for future work.

Here, we numerically examine the detection capability of the three criteria with random thermal states in Fig.~\ref{smfig:thermal}. The detection capability also suffers from exponential decay after a constant period. As $T$ increases, the purity of states decreases, just like the case when $k$ increases in the $\pi_{d,k}$ distribution. This is compatible with the theorems.

\begin{figure}[ht]
 \centering
 \includegraphics[width=8cm]{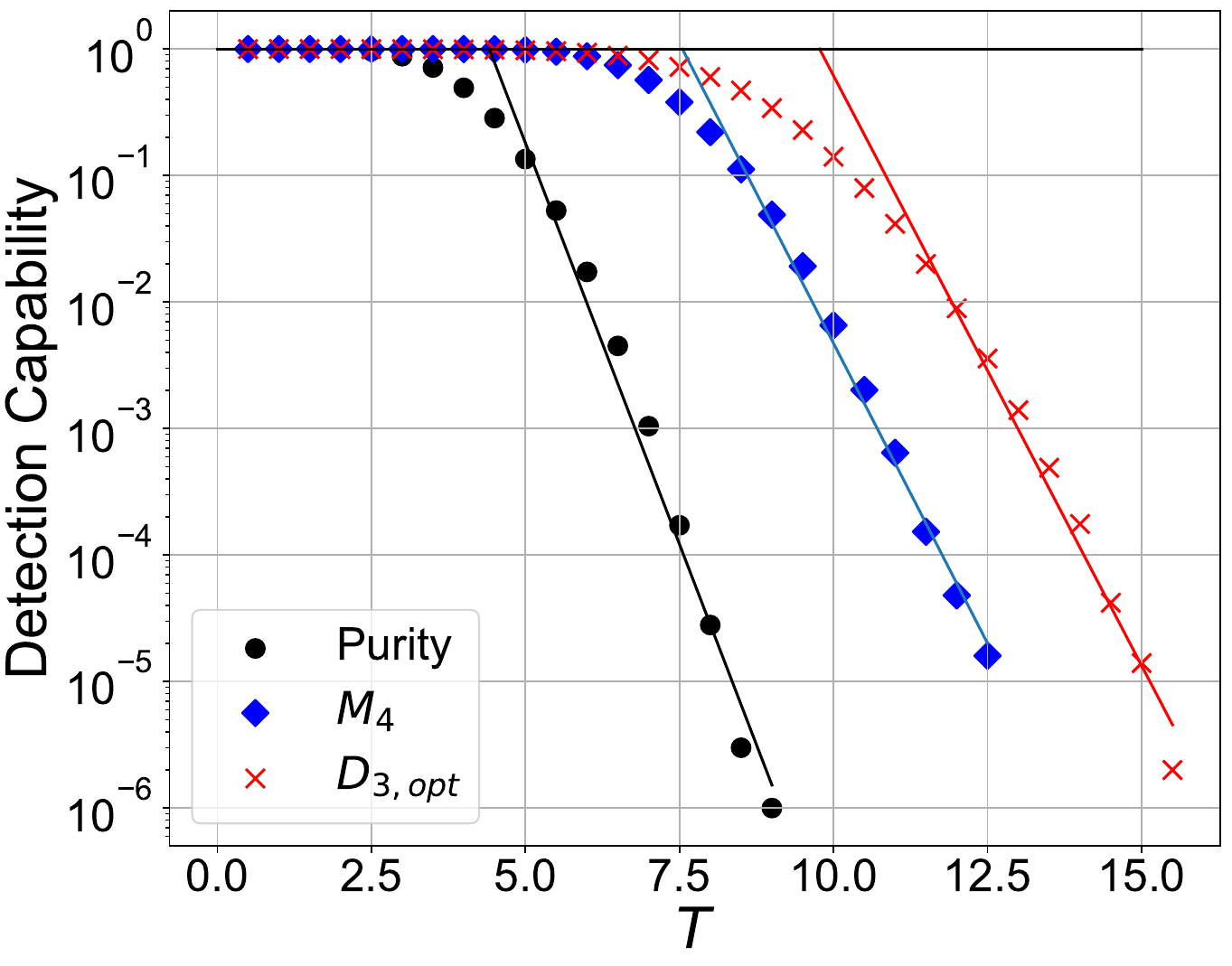}
\caption{This figure represents the detection capability of purity, $M_4$ and $D_{3,\mathrm{opt}}$ criteria respectively with regard to the temperature. We first generate a random Hamiltonian according to the Gaussian orthogonal ensemble, then calculate $\rho=\frac{e^{-\beta H}}{\tr(e^{-\beta H})}$ as the random density matrix, where $\beta=\frac{1}{T}$. Each point is generated through $10^6$ independent experiments.}
\label{smfig:thermal}
\end{figure}

\end{document}